\documentclass[reqno,centertags,11pt]{amsart}
\usepackage{amsmath,amsthm,amscd,amssymb} \usepackage{latexsym}
\usepackage{graphicx}
\usepackage{color}

 \newcommand{\N}{{\mathbb{N}}}
 \newcommand{\R}{{\mathbb{R}}}
 \newcommand{\C}{{\mathbb{C}}}
 
 \newcommand{\Z}{{\mathbb{Z}}}

 {

\newcommand{\beq}{\begin{equation}}
\newcommand{\eeq}{\end{equation}}
\newcommand{\bdm}{\begin{displaymath}}
\newcommand{\edm}{\end{displaymath}} \newcommand{\ba}{\begin{align}}
\newcommand{\ea}{\end{align}} \newcommand{\bpf}{\begin{proof}}
\newcommand{\epf}{\end{proof}}

 \allowdisplaybreaks

\DeclareMathOperator{\Tr}{Tr}

\newtheorem{theorem}{Theorem}[section]
\newtheorem{lemma}[theorem]{Lemma}
\newtheorem{corollary}[theorem]{Corollary}

\theoremstyle{definition}
\newtheorem{definition}[theorem]{Definition}
\theoremstyle{remark}
\newtheorem{remark}[theorem]{Remark}

\begin{document}

\title[Solutions of  Gross-Pitaevskii Equation]{ Solutions of  Gross-Pitaevskii Equation with Periodic Potential in Dimension Three.}
\author[Yu. Karpeshina, Seonguk  Kim,  R.Shterenberg]{Yulia Karpeshina, Seonguk Kim,  Roman Shterenberg}
%


\address{Department of Mathematics, Campbell Hall, University of Alabama at Birmingham,
1300 University Boulevard, Birmingham, AL 35294.}
\email{karpeshi@uab.edu}%

\address{Department of Mathematics, Julian Science and Math Center,
Depauw University,
Greencastle, IN 46135.}%
\email{seongukkim@depauw.edu}

\address{Department of Mathematics, Campbell Hall, University of Alabama at Birmingham,
1300 University Boulevard, Birmingham, AL 35294.}
\email{shterenb@uab.edu}%

\address{}
\email{}

\thanks{Supported in part by NSF-grants DMS- 1814664 (Y.K. and R.S) }

\date{\today}


\maketitle
\begin{abstract} Quasi-periodic solutions of  the Gross-Pitaevskii equation with a periodic potential in dimension three are studied.   It is proven that there is an extensive "non-resonant"  set ${\mathcal G}\subset \R^3$ such that for every $\vec k\in \mathcal G$ there is a  solution  asymptotically close to a plane wave 
$Ae^{i\langle{ \vec{k}, \vec{x} }\rangle}$ as $|\vec k|\to \infty $, given $A$ is sufficiently small.\end{abstract}
\section{Introduction}
Let us consider  the Gross-Pitaevskii equation in dimension three with a periodic potential $V(\vec{x})$ and quasi-periodic boundary condition:
\begin{equation}\label{main equation, 4l>n+1}
-\Delta u(\vec{x})+V(\vec{x})u(\vec{x})+\sigma |u(\vec{x})|^{2}u(\vec{x})=\lambda u(\vec{x}),~\vec{x}\in [0,2\pi]^{3},
\end{equation}
\begin{equation}\label{main condition, 4l>n+1}
\begin{cases}
~u(\vec x)|_{x_s=2\pi}=e^{2\pi it_{s}}u(\vec x)|_{x_s=0},\\
~\frac{\partial}{\partial x_{s}}u(\vec x)|_{x_s=2\pi}=e^{2\pi it_{s}}\frac{\partial}{\partial x_{s}}u(\vec x)|_{x_s=0},\\
s=1,2,3,
\end{cases}
\end{equation}
%
 where $\vec{t}=(t_{1},t_2, t_{3}) \in K=[0,1]^{3}$, $\sigma$ is a real number  and $V(\vec{x})$ is  a trigonometric polynomial with $\int_{Q} V(\vec{x})d\vec{x}=0,$ $Q=[0,2\pi]^{3}$ being the elementary cell of period $2\pi$. More precisely,
\begin{equation}\label{periodic potential, 4l>n+1}
V(\vec{x})=\sum_{q\in \Z^3, 0<|q|< R_0 }v_{q}e^{i\langle{ q, \vec{x} }\rangle},\ \ \ R_0<\infty ,
\end{equation}
$v_{q}$ being Fourier coefficients.

The equation \eqref{main equation, 4l>n+1} is a famous Gross-Pitaevskii equation for Bose-Einstein condensate, see e.g. \cite{PS08}.
%
%
In physics papers, e.g. \cite{KS02}, \cite{LO03}, \cite{YB13}, \cite{YD03}, a big variety of numerical computations for Gross-Pitaevskii equation is made. 
However, they are restricted to the one dimensional case and there is a lack of theoretical considerations even for the case $n=1$. In  paper \cite{KS17} we studied the case   $n=2$. 
%
%

The goal of this paper is to construct asymptotic formulas for $u(\vec{x})$ as $\lambda \to \infty $.
We show that there is an extensive "non-resonant" set ${\mathcal G}\subset \R^3$ such that for every $\vec k\in \mathcal G$ there is a quasi-periodic solution of (\ref{main equation, 4l>n+1}) close to a plane wave 
$Ae^{i\langle{ \vec{k}, \vec{x} }\rangle}$ with $\lambda=\lambda (\vec k, A)$ close to $|\vec k|^{}+\sigma |A|^2$ as $|\vec k|\to \infty $ (Theorem \ref{main theorem 4l>n+1}). We assume $A\in \C$ and $|A|$ is sufficiently small:
\begin{equation}|\sigma| |A|^2<\lambda ^{-\gamma  }, \ \ \gamma >1.
\label{A}
\end{equation}
The quasi-momentum $\vec t$ in \eqref{main equation, 4l>n+1}
is defined by the formula: $\vec k=\vec t+2\pi j$, $j\in \Z^3$.

We show that the non-resonant set $\mathcal G$ has an asymptotically full measure in $\R^3$:
\begin{equation}
\lim _{R\to \infty}\frac{\left| \mathcal G\cap B_R\right|_3}{|B_R|_3}=1, \label{full}
\end{equation}
where $B_R$ is a ball of radius $R$ in $\R^3$ and $|\cdot |_3$ is Lebesgue measure in $\R^3$.

Moreover,
we investigate a set $\mathcal{D}_{}(\lambda, A)$ of vectors $\vec k\in \mathcal G$, corresponding to a fixed sufficiently large $\lambda $ and a fixed $A$. The set $\mathcal{D}_{}(\lambda, A)$,
defined as a level (isoenergetic) set for $\lambda _{}(\vec
k, A)$, 
\begin{equation} 
{\mathcal D} _{}(\lambda,A)=\left\{ \vec k \in \mathcal{G} _{} :\lambda _{}(\vec k, A)=\lambda \right\},\label{isoset} 
\end{equation}
is proven to be a slightly distorted sphere with  
a finite number of holes (Theorem \ref{iso}). For any sufficiently large $\lambda $, it can be described by the formula:
\begin{equation} {\mathcal D}_{}(\lambda, A)=\{\vec k:\vec
k=\varkappa _{}(\lambda, A,\vec{\nu})\vec{\nu},
\ \vec{\nu} \in {\mathcal B}_{}(\lambda)\}, \label{D}
\end{equation}
where ${\mathcal B}_{}(\lambda )$ is a subset of the unit
sphere $S_{2}$. The set ${\mathcal B}_{}(\lambda )$ can be
interpreted as a set of possible directions of propagation for 
almost plane waves. Set ${\mathcal B}_{
}(\lambda )$ has an asymptotically full
measure on $S_{2}$ as $\lambda \to \infty $:
\begin{equation}
\left|{\mathcal B}_{}(\lambda )\right|=_{\lambda \to \infty
}\omega _{2} +O\left(\lambda^{-\delta }\right), \ \ \delta >0,\label{B11}
\end{equation}
here $|\cdot |$ is the standard surface measure on $S_{2}$, $\omega _{2} =|S_{2}|$.
The value $\varkappa _{}(\lambda ,A,\vec \nu )$ in (\ref{D}) is
the ``radius" of ${\mathcal D}_{}(\lambda,A)$ in a direction
$\vec \nu $. The function $\varkappa _{}(\lambda ,A,\vec \nu
)-(\lambda-\sigma |A|^2)^{1/2}$ describes the deviation of ${\mathcal
D}_{}(\lambda,A)$ from the perfect sphere of the radius
$(\lambda-\sigma |A|^2)^{1/2}$ in $\R^3$. It is proven that the deviation is asymptotically
small:
\begin{equation} \varkappa _{}(\lambda ,A, \vec \nu
)=_{\lambda \to \infty} \left(\lambda-\sigma |A|^2\right)^{1/2}+O\left(\lambda ^{-\gamma _1}\right),\ \ \gamma _1>0.
\label{h}
\end{equation}

To prove the results above, we consider the term $V+\sigma|u|^{2}$ in equation (\ref{main equation, 4l>n+1}) as a periodic potential and formally change the nonlinear equation to a linear equation with an unknown potential $V(\vec{x})+\sigma |u(\vec{x})|^{2}$:
\begin{equation*}
-\Delta u(\vec{x})+\big(V(\vec{x})+\sigma |u(\vec{x})|^{2}\big)u(\vec{x})=\lambda u(\vec{x}).
\end{equation*} 
Further, we use known results  for linear Schr\"{o}dinger equations with a periodic potential.
To start with, we consider a linear operator in $L^{2}(Q)$ 
described by the formula:
\begin{equation}H(\vec t)=-\Delta +V,\label{linear oper 4l>n+1}\end{equation}
and quasi-periodic boundary condition (\ref{main condition, 4l>n+1}).
The free operator $H_{0}(\vec t)$, corresponding to $V=0$, has eigenfunctions given by:
\begin{equation}\psi_{j}(\vec{x})=e^{i\langle{ \vec{p}_{j}(\vec t), \vec{x} }\rangle},~~\vec{p}_{j}(\vec t):=\vec t+2\pi j,~j \in \mathbb{Z}^{3},~\vec t \in K,\label{0}\end{equation}
and the corresponding eigenvalues $p_{j}^{}(\vec t):=|\vec{p}_{j}(\vec t)|^{}$. 
Perturbation theory for a linear operator $H(\vec t)$ with a periodic potential $V$ is developed in 
\cite{K97}. It is shown that at high energies, there is an extensive set of generalized eigenfunctions being close to plane waves. Below (See Theorem \ref{Theorem2.1}), we describe this result in details.
Now, we define a map ${\mathcal M}: L^{\infty}(Q) \rightarrow L^{\infty}(Q)$ by the formula:
\begin{align}\label{def of A 4l>n+1}
{\mathcal M}W(\vec{x})=V(\vec{x})+\sigma|u_{\tilde{W}}(\vec{x})|^{2}.
\end{align}
Here, $\tilde{W}$ is a shift of $W$ by a constant such that $\int_{Q}\tilde{W}(\vec{x})d\vec{x}=0$,
\begin{align}\label{tilde W 4l>n+1}
\tilde{W}(\vec{x})=W(\vec{x})-\frac{1}{(2\pi)^{3}}\int_{Q}W(\vec{x})d\vec{x},
\end{align}
and $u_{\tilde{W}}$ is an eigenfunction of the linear operator $-\Delta +\tilde{W}$ 
with the boundary condition (\ref{main condition, 4l>n+1}). 
Next, we consider a sequence $\{W_{m}\}_{m=0}^{\infty}$:
\begin{align}\label{def of successive sequence A 4l>n+1}
W_{0}=V+\sigma |A|^2,\ \ \ {\mathcal M}W_{m}=W_{m+1}.
\end{align}
Note that the sequence is well-defined by induction, since  for each $m=0,1,2,\dots$ and $\vec t$ in a non-resonant set ${\mathcal G}$ described in Section 2, there is an eigenfunction $u_{{m}}(\vec{x})$ corresponding to the potential $\tilde{W}_{m}$:
%
$$H_{{m}}(\vec t)u_{{m}} =\lambda_{{m}}u_{{m}},$$
$$H_{{m}}(\vec t)u_{{m}} :=-\Delta^{}u_{{m}}+\tilde{W}_{m}u_{{m}},$$
 $\lambda_{{m}}$, $u_{m}$ being defined by formal series of the form \eqref{3.2.10}, \eqref{3.2.10a}, \eqref{3.2.13}, \eqref{3.2.14}, \eqref{def of u 2l>n} with $\tilde W_m$ instead of $V$. Those series are proven to be convergent, thus justifying our construction. 
Next, we prove that the sequence $\{W_{m}\}_{m=0}^{\infty}$ is a Cauchy sequence of periodic functions in $Q$ with respect to a norm 
\begin{align}\label{def of star norm 4l>n+1}
\|W\|_{*}=\sum_{q\in \mathbb{Z}^{3}}|w_{q}|,
\end{align}
$w_{q}$ being Fourier coefficients of $W$. 
This implies that 
$W_{m}\rightarrow W$ {with respect to the norm} $\|\cdot\|_{*}$, $W$ {is a periodic function}. 
Further, we show that 
$$u_{{m}}\rightarrow u_{\tilde W} ~\mbox{in}~L^{\infty}(Q),\ \ \  \ \  \lambda_{{m}}\rightarrow \lambda_{\tilde W} ~\mbox{in}~\mathbb{R},$$
where $u_{\tilde W}$, $ \lambda_{\tilde W}$ correspond to the potential ${\tilde W}$ (via \eqref{3.2.10}, \eqref{3.2.10a}, \eqref{3.2.13}, \eqref{3.2.14}, \eqref{def of u 2l>n} with $\tilde W$ instead of $V$).
%
It follows from (\ref{def of A 4l>n+1}) and (\ref{def of successive sequence A 4l>n+1}) that ${\mathcal M}W=W$ and, hence, $u_{}:=u_{\tilde W}$ solves the nonlinear equation 
 (\ref{main equation, 4l>n+1}) with quasi-periodic boundary condition (\ref{main condition, 4l>n+1}). 
  We use the following norm $\|T\|_1$ of an operator $T$ in $l_2(\Z^3)$:
\begin{equation}\|T\|_1=\max _{i}\sum _p|T_{pi}|.\label{norm1} \end{equation}
%


The paper is organized  as follows. In Section 2, we introduce  results for the linear operator $-\Delta +V$ which include the perturbation formulas for an eigenvalue and its spectral projection. In Section 3, we prove  existence of  solutions of the equation \eqref{main equation, 4l>n+1} with boundary condition \eqref{main condition, 4l>n+1}
and investigate their properties. Isoenergetic surfaces are also introduced and described there. Section 4 contains several technical appendices which adjust the results from \cite{K97} to our present setting.
%
%
%
%
%
%
%
%
%

 \section{The Main Results  for the Linear Case.}
In this section we consider the linear operator with a periodic potential in $L_2(\R^3)$:
\begin{equation}\label{H} H=-\Delta +V \end{equation} 
We remind results proven in \cite{K97}. It is well-known that the  spectral analysis of $H$ can be 
reduced to studying the operators $H(\vec t)$,  $\vec t\in K$, where $K$ is the unit cell 
of the dual lattice, $K=[0,1]^3.$ The vector $\vec t$ is being called quasimomentum. 
The operator $H(\vec t)$, $\vec t\in K$, acts in 
$L_2(Q),\  Q=[0,2\pi]^3$. 
Its action is described by  
formula~(\ref{main equation, 4l>n+1}) together with the quasiperiodic conditions.

The operator $H(\vec t)$ has a discrete semi-bounded spectrum $\Lambda (\vec t)$:
$$\Lambda (\vec t)=\cup _{n=1}^{\infty }\lambda _n(\vec t),\  
\lambda _n(\vec t)\to_{n\to \infty }\infty .$$
The spectrum $\Lambda $ of operator $H$ is the union of the spectra $\Lambda (\vec t)$,
$$\Lambda =\cup _{\vec t\in K}\Lambda (\vec t)=\cup _{n\in \N,\vec t\in K}\lambda _n(\vec t).$$
The functions $\lambda _n(\vec t)$ are continuous, so $\Lambda $ has a band 
structure:
$$\Lambda =\cup _{n=1}^{\infty }[q_n,Q_n],\ \ 
 q_n=\min _{\vec t\in K}\lambda _n(\vec t),\  \ 
 Q_n=\max _{\vec t\in K}\lambda _n(\vec t).$$
The eigenfunctions of $H(\vec t)$ and $H$ are simply related. 
If we extend the 
eigenfunctions of all the operators $H(\vec t)$ quasiperiodically (see~(\ref{main condition, 4l>n+1})) 
to $\R^3$, we obtain a complete system of eigenfunctions of the 
operator $H$.
 
Let $H_0(\vec t)$ be the operator corresponding to the zero potential. 
Its eigenfunctions are the plane waves:
\begin{equation}
\exp\{i\langle\vec{p}_j(\vec t),\vec x\rangle\},\ j\in \Z^3,\ \vec{p}_j(\vec t)=\vec{p}_j(0)+\vec t.  
\label{7a}
\end{equation}
The eigenfunction~(\ref{7a}) corresponds to the eigenvalue $p_j^{2}(\vec t)=
|\vec{p}_j(\vec t)| ^{2}$. 
Thus, the spectrum of $H_0$ is equal to
$$\Lambda _0(\vec t)=\{ p_j^{2}(\vec t)\} _{j\in \Z^3}.$$
Using the basis of the eigenfunctions of $H_0(\vec t)$ one can write the matrix 
$H(\vec t)$ in the form
\begin{equation} 
H(\vec t)_{mj}=p_m^{2}(\vec t)\delta _{mj}+v_{m-j}, \label{8}
\end{equation}
where $\delta_{mj}$ is the Kronecker symbol. Of course, the free operator is diagonal
 in this basis.

Note that any $\vec{k}\in \R^3$ can be uniquely represented in the form:
\begin{equation}
\vec{k}=\vec{p}_j(\vec t),\ j\in \Z^3,\ \vec t\in K. \label{y}
\end{equation}
Thus, any plane wave $\exp\{ i\langle\vec{k},\vec x\rangle\}$ can be written in the
 form (\ref{7a}). 
 

In physical literature, the important concept of the isoenergetic surface of the free
 operator is 
used (see e.g. \cite{2r,1r,3r}). 
It is said that a point $\vec t$ belongs to an isoenergetic surface
 $S_0(k)$  of the free operator $H_0$, 
if and only if, the operator $H_0(\vec t)$ has an eigenvalue equal to $k^{2}$, 
i.e., there exists $m\in \Z^3$, such that 
$p_m^{2}(\vec t)=k^{2}$. 
This surface can be obtained as follows: 
the sphere of 
radius $k$ centered at the origin of $\R^3$ is divided into pieces by the dual 
lattice $\{ \vec{p}_m(\vec t)\} _{m\in \Z^3}$, and then all these pieces are transmitted 
into the cell $K$ of the dual lattice. Thus, we obtain the sphere 
``packed into the bag'' $K$.

To describe the main results we  introduce a model operator 
$\hat{H}(\vec t)$. 
First, we define the set $\Gamma (R_0)$. Let us consider $j:j\in \Z^3, 
|j|<R_0$. In this set some of the $j$ are scalar multipliers of 
others. Let us keep from every family of scalar multipliers only the 
minimal representative, i.e., the representative having the minimal length. 
We denote by $\Gamma (R_0)$ the union of these minimal representatives.
In other words, each $j:j\in \Z^3, 
|j|<R_0$ can be uniquely represented in the form $j=mj_0$, where 
$m\in \Z$, $j_0\in \Gamma (R_0)$. It is easy to see that potential $V(\vec x)$ can be written in the form:
\begin{equation} 
V=\sum _{q\in \Gamma (R_0)}V_q,    \label{3.2.2}
\end{equation}
where $V_q$ depends only on the single variable $\langle\vec x,\vec{p}_q(0)\rangle$:
\begin{equation}
V_q(\vec x)=\sum _{|nq|< R_0,n\in \Z}v_{nq}\exp \{in\langle\vec x,\vec{p}_q(0)\rangle\}. \label{3.2.5}
\end{equation}
Further we use this representation of the potential.

Let us consider the following sets in $\Z^3$:
\begin{equation}
\Pi _q(k^{1/5})=\left\{ j:\mid \langle\vec{p}_j(0),\vec{p}_q(0)\rangle\mid <
k^{1/5} ,\right\}\footnote{In fact, there is an auxiliary coefficient in front of $k^{1/5}$, which is equal 
to $1/5$ or $5$. This coefficient arises for technical reasons; we drop it here to describe
 the principal scheme.}
     \label{3.2.6}
\end{equation}

$$
T (k, R_0)=\{ j:\exists q,q'\in \Gamma (R_0), q\neq q':
$$
\begin{equation} 
\mid \langle\vec{p}_{j}(0),\vec{p}_q(0)\rangle\mid <k^{1/5} ,
\mid \langle\vec{p}_{j}(0),\vec{p}_{q'}(0)\rangle\mid 
<k^{3/5}\} \label{3.2.7}
\end{equation}
Let us define a diagonal projection $P_q$ as follows: 
\begin{equation} 
(P_q)_{jj}=\left\{ \begin{array}{ll}1, &\mbox{if $j \in \Pi _q(k^{1/5}
)\setminus T(k, R_0)$;}\\0, &\mbox{otherwise.}\end{array}\right.  \label{3.2.7a}
\end{equation}
 We define the model operator $\hat{H}(\vec t)$ by the formula 
\begin{equation} 
\hat{H}(\vec t)=H_0(\vec t)+\sum _{q\in \Gamma (R_0)}P_qV_qP_q,   \label{3.2.8}
\end{equation}
$H_0$ being the free operator ($V=0$).
Let
\begin{equation}
 \hat W _0=V-\sum _{q\in \Gamma (R_0)}P_qV_qP_q,  \label{3.2.9}
\end{equation}
i.e., $H(\vec t)=\hat{H}(\vec t)+\hat W$. Further, let
\begin{equation}
 \hat{g}_r(k,\vec t)=\frac{(-1)^r}{2\pi ir}\mbox{Tr}\oint _{C_0}((\hat{H}(\vec t)-z)^{-1}\hat W_0)^rdz,
 \label{3.2.10}
 \end{equation}
\begin{equation}
\hat{G}_r(k,\vec t)=\frac{(-1)^{r+1}}{2\pi i}\oint _{C_0}((\hat{H}(\vec t)-z)^{-1}\hat W_0)^r
(\hat{H}(\vec t)-z)^{-1}dz,
\label{3.2.10a} \end{equation}
 $C_0$ being the circle of the radius $k^{-1-\delta }$ about the point $z=k^2$. 
In \cite{K97}  we described  the set $\chi _3(k,V,\delta)\subset S_0(k)$:
 such that for any $t$ of 
 this set  
the operator $\hat{H}(\vec t)$ has a unique eigenvalue $p_j^2(\vec t)$ inside $C_0$, $j$ 
being uniquely determined from the relation 
$p_j^2(\vec t)=k^2$. 
This assertion is stable with respect to $\vec t$: if $\vec t$ is of 
the $(k^{-2-2\delta })$-neighborhood of $\chi _3(k,V,\delta )$, then 
the operator $\hat{H}(\vec t)$ has a unique eigenvalue $p_j^2(\vec t)$ inside $C_0$, $j$ 
being uniquely determined from the relation $p_j^2(\vec t)\in \varepsilon (k,\delta )
\equiv [k^2-k^{-1-\delta },k^2+k^{-1-\delta }]
.$
The spectral projection
 ${\mathbb E}_j$ (the same as for the free operator) corresponds to this eigenvalue.

For the operator 
$$\hat{A}=(\hat{H}(\vec t)-z)^{-1/2}\hat W _0
(\hat{H}(\vec t)-z)^{-1/2}$$ we have:
\begin{equation}
\|\hat{A}\|<k^{2\delta },\ \ \|\hat{A}^3\|<k^{-1/5+21\delta }. \label{3.2.11}
\end{equation}
\begin{theorem} \label{Theorem2.1} Suppose $\vec t$ is in the $(k^{-2-2\delta })$-neighborhood of 
the nonsingular set $\chi_3(k,V,\delta )$, $0<\delta <1/200$. 
Then for sufficiently 
large $k$,  $k>k_0(V,\delta )$, in 
the interval 
$\varepsilon (k,\delta )\equiv [k^{2}-k^{-1-\delta },k^{2}+k^{-1-\delta }]$
there exists a single eigenvalue of the operator $H$. It is given by series:
\begin{equation}
\lambda (\vec t)=p_j^{2}(\vec t)+\sum _{r=2}^{\infty}\hat g_r(k,\vec t)
                                                    \label{3.2.13}
\end{equation}
where $j$
is uniquely determined from the relation 
$p_j^{2}(\vec t)\in \varepsilon (k,\delta )$. 
The spectral projection corresponding to $\lambda (\vec t)$ is determined by the series:
\begin{equation}  
E(\vec t)={\mathbb E}_j+\sum _{r=1}^{\infty}\hat G_r(k,\vec t)
 \label{3.2.14}
\end{equation}
%
which converges in the trace class $\hbox{\bf{S}}_1$.

For the functions $\hat g_r(k,\vec t)$ and the operator-valued functions $\hat G_r(k,\vec t)$ the estimates
\begin{equation} 
\mid \hat g_r(k,\vec t)\mid <k^{-1-\delta -r/20},  \label{3.2.15}
\end{equation}
\begin{equation} 
\| \hat G_r(k,\vec t)\|_{\bf{S}_1} <k^{-r/20}  \label{3.2.16}
\end{equation}
hold.
%
%
%
\end{theorem}

It turns out that estimates (\ref{3.2.15}) and (\ref{3.2.16}) can be 
improved when $r<k^{\delta }R_0^{-1}$. 
\begin{lemma} Under the conditions of Theorem \ref{Theorem2.1} with 
$r<k^{\delta }R_0^{-1}$ we have:
\begin{equation}
|\hat g_r(k,\vec t)|<\hat{v}r^{2}(\hat{v}k^{-1+3\delta })^{r-1},  \label{3.2.21}
\end{equation}
\begin{equation}
\|\hat G_r(k,\vec t)\|<(\hat{v}k^{-1+3\delta })^{r},  \label{3.2.22}
\end{equation}
\begin{equation}
\|\hat G_r(k,\vec t)\|_{\bf{S}_1}<(rR_0)^{3}(\hat{v}k^{-1+3\delta })^{r},  \label{3.2.23}
\end{equation}
\begin{equation}
|\hat g_2(k,\vec t)|<\hat{v}^2R_0^{-1}k^{-2+6\delta }. \label{3.2.24}
\end{equation}
Here and below:
$$\hat{v}\equiv c_0(\max _{|q|<R_0}|v_q|)R_0^3, \ \ c_0\neq c_0(k,V).$$
 
The operator $\hat G_r(k,\vec t)$, $r\in \N$, is nonzero only on the finite-dimensional subspace 
$(\sum _{i\in \Z^3, |i-j|<rR_0}E_i)l_2^{3}$.
\end{lemma}
\begin{corollary} The perturbed eigenvalue and its spectral 
projection satisfy the following estimates:
\begin{equation} 
\mid \lambda (\vec t)-p_j^{2}(\vec t)\mid \leq c\hat{v}^2(\hat{v}+R_0^{-1})
k^{-2+6\delta },
\label{3.2.25}
\end{equation}
\begin{equation}
\|E(\vec t)-{\mathbb E}_j\|_{\bf{S}_1}\leq c\hat{v}R_0^{3}k^{-1+3\delta }.
\label{3.2.26}
\end{equation}
\end{corollary} 
Let us introduce the notations:
\begin{equation}
T(m):=\frac{\partial^{|m|}}{\partial t_1^{m_1}\partial t_2^{m_2}\partial t_3^{m_3}},\ \ |m|=m_1+m_2+m_3,\ \ m!=m_1!m_2!m_3!,
\end{equation}
with $0\leq |m|<\infty,\ \ T(0)f:=f$. 
\begin{lemma} Under the conditions of Theorem \ref{Theorem2.1} the functions $\hat{g}_r
(k,\vec t)$ and the 
operator-valued functions $\hat{G}_r(k,\vec t)$ 
depend analytically on $\vec t$ in the complex 
$(k^{-2-2\delta })$-neighborhood of each simply connected component of the 
nonsingular set $\chi _3(k,V ,\delta )$. They satisfy the estimates:
\begin{equation} 
\mid T(m) \hat{g}_r(k,\vec t)\mid <m!k^{-1-\delta -r/20+2(1+\delta )|m|},  
\label{3.2.27}
\end{equation}
\begin{equation} 
\| T(m) \hat{G}_r(k,\vec t)\| <m!k^{-r/20+2(1+\delta )|m|},  \label{3.2.28}
\end{equation}
%
%
If $r<k^{\delta }R_0^{-1}$, then:
\begin{equation}
|T(m)\hat g_r(k,\vec t)|<m!(c_0k^{3\delta })^{|m|}
\hat{v}r^{2}(\hat{v}k^{-1+3\delta })^{r-1},  \label{3.2.31}
\end{equation}
\begin{equation}
\|T(m)\hat G_r(k,\vec t)\|<m!(c_0k^{3\delta })^{|m|}
(\hat{v}k^{-1+3\delta })^{r},  \label{3.2.32}
\end{equation}
\begin{equation}
\|T(m)\hat G_r(k,\vec t)\|_{\bf{S}_1}<m!(c_0k^{3\delta })^{|m|}
(rR_0)^{3}(\hat{v}k^{-1+3\delta })^{r},  \label{3.2.33}
\end{equation}
\begin{equation}
|\hat g_2(k,\vec t)|<m!(c_0k^{3\delta })^{|m|}
\hat{v}^2R_0^{-1}k^{-2+6\delta },  \label{3.2.34}
\end{equation}
$$ c_0\neq c_0(k,V).$$
\end{lemma} 
\begin{corollary} \label{Cor1} The function $\lambda (\vec t)$ and the operator-valued function $E(\vec t)$ 
depend analytically on $\vec t$ in the complex 
$(k^{-2-2\delta })$-neighborhood of each simply connected component of the 
nonsingular set $\chi _3(k,V ,\delta )$. They admit the estimates:
\begin{equation}
\mid T(m) (\lambda (\vec t)-p_j^{2}(\vec t))\mid \leq cm!\hat{v}^2(\hat{v}+R_0^{-1})
k^{-2+6\delta +2(1+\delta )|m|},
\label{3.2.35}
\end{equation}
\begin{equation} 
\|T(m) (E(\vec t)-{\mathbb E}_j)\|_{\bf{S}_1}\leq cm!\hat{v}R_0^{3}k^{-1+3\delta +2(1+\delta )|m|}.
\label{3.2.36}
\end{equation}
If $|m|< k^{\delta }(60R_0)^{-1}$, then the estimates can be improved:
\begin{equation} 
\mid T(m) (\lambda (\vec t)-p_j^{2}(\vec t))\mid \leq cm!(c_0k^{3\delta })^{|m|}
\hat{v}^2(\hat{v}+R_0^{-1})k^{-2+6\delta },
\label{3.2.37}
\end{equation}
\begin{equation}
\|T(m) (E(\vec t)-{\mathbb E}_j)\|_{\bf{S}_1}\leq cm!(c_0k^{3\delta })^{|m|}
\hat{v}R_0^{3}k^{-1+3\delta }.
\label{3.2.38}
\end{equation}
\end{corollary} 

\begin{corollary} \label{def1}There is a one-dimensional space of Bloch eigenfunctions $u_0$ corresponding to the projection $E(\vec t)$ given by  \eqref{3.2.14}.
They are given by the formula:
\begin{align}\label{def of u 2l>n}
 u_0(\vec{x})&=
 A\sum_{m \in \mathbb{Z}^{3}}E(\vec t)_{mj}e^{i\langle{ \vec{p}_{m}(\vec t), \vec{x} }\rangle}\\
\nonumber &=Ae^{i\langle{ \vec{p}_{j}(\vec t), \vec{x} }\rangle}\Bigg(1+\sum_{q\neq 0}\frac{v_{q}}{p_{j}^{}(\vec t)-p_{j+q}^{}(\vec t)}e^{i\langle{\vec{p}_{q}({0}), \vec{x} }\rangle}+\cdot\cdot\cdot\Bigg),~~j,q \in \mathbb{Z}^{3},\ A\in \C.
\end{align}
\end{corollary}
Let $\tilde \chi _3(k,V,\delta)\subset S(k)$ be the image of $\chi _3(k,V,\delta)\subset S_0(k)$ on the sphere $S(k)$:
\begin{equation}
\tilde \chi _3(k,V,\delta)=\{\vec p_j(\vec t)\in S(k):\  \vec t \in  \chi _3(k,V,\delta)\}.\end{equation}
Note that $\tilde \chi _3(k,V,\delta)$ is well-defined, since $ \chi _3(k,V,\delta)$ does not contain self intersections of $S_0(k)$.
Let $\mathcal B(\lambda )\subset S_{2}$ be the set of directions corresponding to the nonsingular set $\tilde \chi _3(k,V,\delta)$: \begin{equation} \label{formulaB}
\mathcal B(\lambda)=\big\{\vec \nu \in S_{2}: k\vec \nu \in \tilde \chi _3(k,V,\delta)\big\}, \ k^{}=\lambda. \end{equation}
The set $\mathcal B(\lambda)$ can be interpreted as a set of possible directions of propagation for almost plane waves \eqref{def of u 2l>n}.
We define the non-resonance set $\mathcal G\subset \R^3$ as the union of all $ \tilde \chi _3(k,V,\delta)$:
\begin{equation} \label{G}
\mathcal G=\bigcup\limits _{k>k_0(V, \delta )}\tilde  \chi _3(k,V,\delta)
\end{equation}
Further we denote vectors of $\mathcal G$ by $\vec k$. Formulas \eqref{formulaB}, \eqref{G} yield:
\begin{equation} \label{G+}
\mathcal G=\big\{\vec k=k\vec \nu: \vec \nu \in \mathcal B(k^{}), \ k>k_0(V,\delta )\big\}. \end{equation}
Since any vector $\vec k$  can be written as $\vec k=\vec p_j(t)$ in a unique way, formula \eqref{G} yields:
\begin{equation} \label{G*}
\mathcal G=\big\{\vec p_j(\vec t):\ \vec t \in  \chi _3(k,V,\delta),\ \mbox{where } k=p_j(\vec t),\ k>k_0(V,\delta )\big\}. \end{equation}
Let $\lambda (\vec k)$ be defined by \eqref{3.2.13}, where $\vec k=\vec p_j(\vec t)$.

Next, we describe isoenergetic surfaces for  the operator \eqref{H}. The set $\mathcal{D}_{}(\lambda)$,
defined as a level (isoenergetic) set for $\lambda _{}(\vec
k)$, \begin{equation} {\mathcal D} _{}(\lambda )=\left\{ \vec k \in \mathcal{G} _{} :\lambda _{}(\vec k)=\lambda \right\}.\label{isoset-lin} \end{equation}
\begin{lemma}\label{L:2.12a} For any sufficiently large $\lambda $, $\lambda >k_0(V,\delta )^{}$, and for every
$\vec{\nu}\in\mathcal{B} (\lambda)$, there is a
unique $\varkappa =\varkappa (\lambda ,\vec{\nu})$ in the
interval $$
I:=[k-k^{-2-2\delta},k+k^{-2-2\delta}],\quad
k^{2}=\lambda , $$ such that \begin{equation}\label{2.70}
\lambda (\varkappa \vec{\nu})=\lambda . \end{equation}
Furthermore, $|\varkappa - k| \leq ck^{-3+5\delta}$.
\end{lemma} 
The Lemma easily follows from \eqref{3.2.27} for $|m|=1$.

\begin{lemma} \label{L:2.13a} \begin{enumerate} \item For any sufficiently
large $\lambda $, $\sqrt \lambda >k _0(V,\delta )$, the set $\mathcal{D}(\lambda )$, defined by \eqref{isoset-lin} is a distorted
sphere with holes; it is described by the formula:
\begin{equation} 
{\mathcal D}_{}(\lambda)=\{\vec k:\vec
k=\varkappa _{}(\lambda,\vec{\nu})\vec{\nu},
\ \vec{\nu} \in {\mathcal B}_{}(\lambda)\},\label{May20} 
\end{equation} where $\varkappa (\lambda,\vec \nu)=k+h (\lambda, \vec \nu)$ and $h (\lambda, \vec \nu)$ obeys the
inequalities:
\begin{equation}\label{2.75a}
|h|<ck^{-3+5\delta },\quad
\left|\nabla _{\vec \nu }h\right| <
ck^{-3+8\delta}.
\end{equation}

\item The measure of $\mathcal{B}(\lambda)\subset S_{2}$ satisfies the
estimate \eqref{B11}.


\item The surface $\mathcal{D}(\lambda)$ has the measure that is
asymptotically close to that of the whole sphere of the radius $k$ in the sense that
\begin{equation}\label{2.77}
\bigl |\mathcal{D}(\lambda)\bigr|\underset{\lambda \rightarrow
\infty}{=}\omega _{2}k^{2}\bigl(1+O(k^{-\delta})\bigr),\quad \lambda =k^{2}.
\end{equation}
\end{enumerate}
\end{lemma}
The proof is based on Implicit Function Theorem.

\section{Proof of The Main Result}
First, we prove that  $\{W_{m}\}_{m=0}^{\infty}$ in \eqref{def of successive sequence A 4l>n+1} is a Cauchy sequence with respect to  the norm defined by (\ref{def of star norm 4l>n+1}).
%
%
%
Further we  need the following  obvious properties of norm  $\|\cdot\|_{*}$:
\begin{align}\label{properties 1 2l>n}
\|f\|_{*}=\|\bar{f}\|_{*},\ \
\|\Re(f)\|_{*}\leq\|f\|_{*},\ \
\|\Im(f)\|_{*}\leq\|f\|_{*},\ \
\|fg\|_{*}\leq \|f\|_{*}\|g\|_{*},
\end{align}
where $\Re(f)$ and $\Im(f)$ are real and imaginary part for $f$, respectively.
%

%
%
Let  $W_m$, $\tilde W_m$ be defined by and  \eqref{def of  A 4l>n+1} --\eqref{def of successive sequence A 4l>n+1}, $\tilde W_0=V$, $$\hat W_m=\tilde W_m-\sum _{q\in \Gamma (R_0)}P_qV_qP_q.$$
\begin{lemma}\label{main lemma  1 2l>n} Let $|\sigma ||A|^2<k^{-1-6\delta }$,  $k$ being sufficiently large $k>k_{1}(V,\delta )$.
The following inequalities hold for any  $m=1,2,\dots$:
\begin{align}
\|\tilde{W}_{m}-V\|_{*}\leq 8|\sigma||A|^{2}k^{-1+4\delta}, \label{mmm}
\end{align}
\begin{align}
\|W_{m}-W_{m-1}\|_{*}\leq 4|\sigma||A|^{2}k^{-1+4\delta }(2^6|\sigma ||A|^2k^{1+5\delta })^{m-1}, \label{mm}
\end{align}
Each operator $H_0(t)+W_m$ has a unique simple eigenvalue in the interval $(k^2-k^{-1-\delta}, k^2+k^{-1-\delta})$. The corresponding spectral projectors $E_m$ satisfy:
\begin{align}\label{estimate of difference of E-a}
 \|E_{{m}}(\vec t)-E_{{m-1}}(\vec t)\|_1 
\leq &4|\sigma||A|^{2}k^{14\delta}(2^6|\sigma ||A|^2k^{1+5\delta})^{m-1},
\end{align}
where $E_0=E$ is given by formula \eqref{3.2.14}.
\end{lemma}
\begin{corollary}\label{cauchy sequence 2l>n} 
There is a periodic function $W $   such that $W_{m}$ converges to $W$ with respect to the norm $\|\cdot\|_{*}$:
\begin{align}
\|W-W_{m}\|_{*}\leq 8|\sigma||A|^{2}k^{-1+4\delta}(2^6|\sigma ||A|^2k^{1+5\delta})^{m}. \label{mm+}
\end{align}
\end{corollary}
\begin{remark} The spectral projectors $E_m$ admit series expantions similar to \eqref{3.2.14}, see \eqref{new}.\end{remark}
\begin{proof}[Proof of Lemma 3.1] First, let us remind that $\hat H$ is given by \eqref{3.2.8}.
It follows from the definition of $\chi_3(k,V,\delta)$  (see Appendix 1 for details) that: 
\begin{align}\label{l}
\max_{z\in C_{0}}\Big\|(\hat H (\vec t)-z)^{-1/2}\Big\|_1<k^{(1+\delta)/2}, \ \ z\in C_{0},
\end{align}
where  $\|\cdot \|_1$ is defined by \eqref{norm1}. We will also use the estimate in $\hbox{\bf S}_2$ norm (see \cite{K97}):
\begin{align}\label{lS2}
\max_{z\in C_{0}}\Big\|(\hat H (\vec t)-z)^{-1}\Big\|_{\hbox{\bf S}_2}<k^{1+\delta}, \ \ z\in C_{0},
\end{align}
Let
\begin{equation}B_0(z)=(\hat H(\vec t)-z)^{-\frac{1}{2}}\hat W _0(\hat H(\vec t)-z)^{-\frac{1}{2}},
  \label{M17b0} \end{equation}
  $\hat W _0$ being given by \eqref{3.2.9}.
Considerations in \cite{K97} yield (see Appendix 2 for details):
\begin{equation}
\|B_0\|_1<k^{2\delta },\ \ \|B_0^3\|_1<k^{-1/5+21\delta }. \label{3.2.11b}
\end{equation}
Furthermore,
\begin{equation}
\|\hat G_r(k,\vec t)\|_1<k^{-(1-4\delta )r}, \ \mbox{when}\  r<k^{\delta }R_0;
\label{Ap3a}
\end{equation}
\begin{equation}
\| \hat G_{r}(k,\vec t)\|_{1}<
k^{-r/20},  \ \mbox{when}\  r\geq  k^{\delta }R_0,  \label{Ap3b}
\end{equation}
 these estimates being just slightly different from \eqref{3.2.16},  \eqref{3.2.23} (see Appendix 3 for explanations).
 Summing \eqref{Ap3a}, \eqref{Ap3b} over $r$, we get that the series \eqref{3.2.14} converges in $\|\cdot \|_1$ and
\begin{equation}\label{jj}
 \|E_{{}}(\vec t)\|_1
= 1+O(k^{-1+4\delta}).
\end{equation}
Next, we prove \eqref{mmm}, \eqref{mm}  for $m=1$. Let us consider the function \eqref{def of u 2l>n}  written in the form
\begin{align}\label{expression u 2l>n 2}
u_0(\vec{x})=\psi _0(\vec{x})e^{i\langle{ \vec{p}_{j}(\vec t), \vec{x} }\rangle},
\end{align}
where 
\begin{align}\label{change cordinator 2l>n*}
\psi _0(\vec{x})=A\sum_{q\in \mathbb{Z}^{3}}E(\vec t)_{j+q,j}e^{i\langle{ \vec{p}_{q}(\vec{0}), \vec{x} }\rangle},
\end{align}
is  called the periodic part of $u_0$. 

 It follows from (\ref{def of A 4l>n+1}), (\ref{def of successive sequence A 4l>n+1}) and  (\ref{properties 1 2l>n}) that 
\begin{align}\label{change cordinator 2 2l>n}
\nonumber \big\|W_{1}-W_{0}\big\|_{*} =&
|\sigma| \big\||u_{0}|^{2}-|A|^2\big\|_{*}=|\sigma| \big\||\psi_{0}|^{2}-|A|^2\big\|_{*}\\
\leq& \nonumber|\sigma| \big\||\psi_{0}|^{2}-|A|^2+2i\Im(\bar A\psi_{0})\big\|_{*}=\nonumber |\sigma| \big\|(\psi_{0}-A)(\bar{\psi}_{0}+\bar A)\big\|_{*}\\
\leq& |\sigma|\big\|\psi_{0}-A\big\|_{*}\big\|\bar{\psi}_{0}+\bar A\big\|_{*}.
\end{align}
Using \eqref{change cordinator 2l>n*} and  \eqref{3.2.14}, we obtain:
\begin{align} \label{est psiv 2l>n*}
\nonumber \|\psi_{0}-A\|_{*}
 \leq &\Big|AE(\vec t)_{jj}-A\Big|+|A|\sum_{q\in\mathbb{Z}^{3}\setminus\{0\}}\Big|E_{}(\vec t)_{j+q,j}\Big| 
 \\
\leq& |A| \sum_{r=1}^{\infty } \| \hat G_{r}(k,\vec t)\|_{1}
\end{align}

 %
The estimates \eqref{Ap3a} and \eqref{Ap3b} yield:
\begin{equation} \label{t}
 \|\psi_{0}-A\|_{*}
<
2|A|k^{-1+4\delta}.
\end{equation}
It follows:
\begin{align}\label{est bar psiv 2l>n*}
\|{\psi}_{0}\|_{*}=\|\bar{\psi}_{0}\|_{*}\leq |A|+ O\big(|A|k^{-1+4\delta }\big).
\end{align}
Using (\ref{change cordinator 2 2l>n}),  (\ref{t}) and (\ref{est bar psiv 2l>n*}), we get 
\begin{align*}
&\|W_{1}-W_{0}\|_{*}
\leq 2 |\sigma ||A|^2k^{-1+4\delta },
\end{align*}
i.e.  \eqref{mm} for $m=1$. Since $\|\tilde{W}_{1}-V\|_{*}=\|\tilde{W}_{1}-\tilde{W}_{0}\|_{*}\leq \|W_{1}-W_{0}\|_{*}$, we have:
\begin{align}\label{54.5}
&\|\tilde{W}_{1}-V\|_{*} \leq 2|\sigma||A|^2k^{-1+4\delta },
\end{align}
i.e. \eqref{mmm} for $m=1$.

 Now, we use  mathematical induction to prove \eqref{mmm} and \eqref{mm} simultaneously.
%
Suppose that for all $1\leq s\leq m-1$,
\begin{equation}
\|\tilde{W}_{s}-V\|_{*} \leq 8|\sigma||A|^{2}k^{-1+4\delta}, \label{M19b a}
\end{equation}
\begin{equation}
\|W_{s}-W_{s-1}\|_{*} \leq 4|\sigma||A|^{2}k^{-1+4\delta}(2^6|\sigma ||A|^2k^{1+5\delta} )^{s-1}. \label{M19b}
\end{equation}
%
%
%
Let
\begin{equation}B_s(z)=(\hat H(\vec t)-z)^{-\frac{1}{2}}\hat W_s (\hat H(\vec t)-z)^{-\frac{1}{2}},
 \label{M17b} \end{equation}
 $\hat H$ being the same for all $s$, see \eqref{3.2.8}, and
 \begin{equation}
 \hat W_s=\tilde W_s-\sum _{q\in \Gamma (R_0)}P_qV_qP_q. \label{3.2.9s} \end{equation}
Using    \eqref{M17b} and \eqref{l}, we easily obtain:
\begin{equation} \label{M17cca}
\|B_s-B_{s-1}\|_1 \leq k^{1+\delta }\|\hat W_s-\hat W_{s-1}\|_*=k^{1+\delta }\|\tilde W_s-\tilde W_{s-1}\|_*\leq k^{1+\delta }\|W_s-W_{s-1}\|_*.
\end{equation}
Taking into account  \eqref{M19b}, we arrive at
\begin{equation} \label{M17cc}
\|B_s-B_{s-1}\|_1 \leq 4|\sigma||A|^{2} k^{5\delta }(2^6|\sigma ||A|^2k^{1+5\delta} )^{s-1},\ \  z \in C_0.
\end{equation}
In particular,
\begin{equation} \label{M17ccb}
\|B_s-B_{s-1}\|_1 \leq 4k^{-1-\delta}.
\end{equation}
By  \eqref{3.2.11b} and  and \eqref{M17cc}:
\begin{equation} \label{M17c}
\|B_s(z)\|_1 \leq
 k^{2\delta }+8|\sigma||A|^{2} k^{5\delta }\leq 
2 k^{2\delta },\ \  z \in C_0, {\mbox when } \ \sigma |A|^2<k^{-1-6\delta},
\end{equation}
\begin{equation} \label{M17ch}
\|B_s(z)^3\|_1 \leq k^{-1/5+21\delta }+8 |\sigma||A|^{2} k^{9\delta }<2k^{-1/5+21\delta },\ \  z \in C_0, {\mbox when } \ \sigma |A|^2<k^{-1-6\delta},
\end{equation}
for any  $1\leq s\leq m-1$.
By analogy with \eqref{3.2.10a}, we set:
\begin{equation}
\hat{G}_{s,r}(k,t)=\frac{(-1)^{r+1}}{2\pi i}\oint _{C_0}(\hat{H}(t)-z)^{-1/2}B^r_s
(\hat{H}(t)-z)^{-1/2}dz.
\label{3.2.10ab} \end{equation}
It follows from \eqref{l} and \eqref{M17c},  \eqref{M17ch} that 
\begin{equation}\label{estimate of difference of G}
\|\hat G_{{s},r}(k,\vec t)\|_1\leq 2^rk^{4\delta }k^{-(\frac{1}{5}-21\delta )\left[\frac{r-1}{3}\right]}.
\end{equation}
Therefore,
\begin{equation}E_s={\mathbb E}_j+\sum _{r=1}^{\infty}\hat G_{{s},r}(k,\vec t), \label{new}
\end{equation}
Here the series converges in $\|\cdot \|_1$ norm, $E_0=E$, given by \eqref{3.2.14}.
Next, we note that
\begin{align}\label{estimate of difference of B}
\nonumber& \max_{z\in C_{0}}\|B_{s}^r(z)-B^r_{s-1}(z)\|_1\\
\nonumber\leq& \sum\limits_{j=0}^{r-1}\max_{z\in C_{0}}\|B_{s}(z)-B_{s-1}(z)\|_1\left\|\mathcal P_j\left(B_{s-1}, B_{s}-B_{s-1}\right)\right\|_1 \left\|\mathcal P_{r-1-j}\left(B_{s-1}, B_{s}-B_{s-1}\right)\right\|_1\\
\end{align}
where $\mathcal P_j\left(B_{s-1}, B_{s}-B_{s-1}\right)$ is the homogeneous polynomial of $B_{s-1}$ and $B_{s}-B_{s-1}$ of order $j$. The estimates \eqref{M17ccb}--\eqref{M17ch}
yeild:
$$\left\|\mathcal P_j\left(B_{s-1}, B_{s}-B_{s-1}\right)\right\|_1<2^{j}k^{-(\frac{1}{5}-21\delta )\left[\frac{j}{3}\right]+4\delta }.$$
Taking into account \eqref{M17cca}, we get:
\begin{equation}\label{B}
 \max_{z\in C_{0}}\|B_{s}^r(z)-B^r_{s-1}(z)\|_1<2^rk^{1 +\delta }\|\tilde{W}_{s}-\tilde{W}_{s-1}\|_{*}k^{-(\frac{1}{5}-21\delta )\left[\frac{r-3}{3}\right]+8\delta }.\end{equation}
Considering \eqref{l}, we obtain:
\begin{equation}\label{estimate of difference of G}
\|\hat G_{{s},r}(k,\vec t)-\hat G_{{s-1},r}(k,\vec t)\|_1\leq 2^rk^{1+9\delta }\|\tilde{W}_{s}-\tilde{W}_{s-1}\|_{*}k^{-(\frac{1}{5}-21\delta )\left[\frac{r-3}{3}\right]}.
\end{equation}
Summing the last estimate over $r$, we arrive at
%
%
%
\begin{align}\label{estimate of difference of E}
 \nonumber \|E_{{s}}(\vec t)-E_{{s-1}}(\vec t)\|_1 
\leq &\nonumber\sum_{r=1}^{\infty}
\|\hat G_{{s},r}(k,\vec t)-\hat G_{{s-1},r}(k,\vec t)\|_1\\
\leq &k^{1+10\delta }\|\tilde{W}_{s}-\tilde{W}_{s-1}\|_{*}.
\end{align}
%
%
%
%
%
Summing \eqref{estimate of difference of E} over $s$ and considering \eqref{M19b}, \eqref{jj},we obtain:
\begin{equation}\label{estimate of E}
 \|E_{{s}}(\vec t)\|_1
= 1+o(1), \ \ 1\leq s\leq m-1.
\end{equation}
Let, by analogy with \eqref{def of u 2l>n},
\begin{align}\label{expression u 2l>n 1}
u_{{s}}(\vec{x}):=A\sum\limits_{m \in \mathbb{Z}^{3}}E_{{s}}(\vec t)_{m,j}e^{i\langle{ \vec{p}_{m}(\vec t), \vec{x} }\rangle}, 
\end{align}
where $E_{{s}}(\vec t)$ is the spectral projection \eqref{3.2.14} for the potential $\tilde{W}_{s}$.
Obviously,
\begin{align}\label{expression u 2l>n 2}
u_{{s}}(\vec{x})=\psi_{{s}}(\vec{x})e^{i\langle{ \vec{p}_{j}(\vec t), \vec{x} }\rangle},
\end{align}
where the function, 
\begin{align}\label{change cordinator 2l>n}
\psi_{{s}}(\vec{x})=A\sum_{q\in \mathbb{Z}^{3}}E_{{s}}(\vec t)_{j+q,j}e^{i\langle{ \vec{p}_{q}(\vec{0}), \vec{x} }\rangle},
\end{align}
is   the periodic part of $u_{{s}}$. Clearly,
\begin{equation}\|\psi_{{s}}\|_*\leq  |A|\|E_{{s}}(\vec t)\|_1.\label{M19a}
\end{equation}
Next, considering as in \eqref{change cordinator 2 2l>n},  we obtain:
\begin{align}
\big\|W_{m}-W_{m-1}\big\|_{*}
\leq |\sigma|\big\|\psi_{{m-1}}-\psi_{{m-2}}\big\|_{*}\big\|\bar{\psi}_{{m-1}}+\bar{\psi}_{{m-2}}\big\|_{*},\label{**}
\end{align}
and,
 hence, by \eqref{change cordinator 2l>n},
 \begin{equation} \label{W-m}
 \|W_{m}-W_{m-1}\|_{*}\leq |\sigma ||A|^2\|E_{{m-1}}(\vec t)-E_{{m-2}}(\vec t)\|_1
\left(\|E_{{m-1}}(\vec t)\|_1+
\|E_{{m-2}}(\vec t)\|_1\right).
\end{equation}
%
Using  (\ref{estimate of E}) and (\ref{estimate of difference of E}),  we obtain
\begin{align}
\|W_{m}-W_{m-1}\|_{*}
\leq &8|\sigma||A|^2 k^{1+5\delta }\|\tilde{W}_{m-1}-\tilde{W}_{m-2}\|_{*}. \label{***}\end{align}
Considering  $\|\tilde{W}_{m-1}-\tilde{W}_{m-2}\|_{*}\leq \|{W}_{m-1}-{W}_{m-2}\|_{*}$ and using \eqref{M19b} for $s=m-1$, we arrive at the estimate:
\begin{align} \label{main results 123}
 \|W_{m}-W_{m-1}\|_{*}\leq &8|\sigma||A|^2 k^{1+5\delta}4|\sigma||A|^{2}k^{-1+4\delta }\big(2^6|\sigma||A|^2k^{1+5\delta}\big)^{m-2}\\
\nonumber \leq & 4|\sigma||A|^{2}k^{-1+4\delta}\big(2^6|\sigma||A|^2k^{1+5\delta}\big)^{m-1}, 
\end{align}
when $k>k_{1}(V,\delta)$.
Further, \eqref{main results 123} and \eqref{54.5} enable the estimate
\begin{align*} 
\nonumber \|\tilde{W}_{m}-V\|_{*}\leq& \|\tilde{W}_{m}-\tilde{W}_{m-1}\|_{*}+\|\tilde{W}_{m-1}-\tilde{W}_{m-2}\|_{*}+\cdots+\|\tilde{W}_{1}-V\|_{*}\\
\leq &8|\sigma||A|^{2}k^{-1+4\delta},
\end{align*}
which completes the proof of \eqref{mmm} and \eqref{mm}. Using \eqref{estimate of difference of E}, we obtain \eqref{estimate of difference of E-a}.
\end{proof}

\begin{lemma} \label{lemma3 2l>n}
Suppose $\vec t$ belongs to the $(k^{-2-2\delta })$-neighborhood in $K$ 
of the 
non-resonant set $\chi _3(k,V,\delta )$. Then for every sufficiently 
large $k>k_{1}(V,\delta)$ and every $A\in \C: |\sigma||A|^2 <k^{-1-6\delta }$,  the sequence $E_m(\vec t)$ converges with respect to $\|\cdot \|_1$ to a one-dimensional spectral projection $E_{\tilde W}(\vec t)$ of $H_0(t)+\tilde W$:
\begin{equation}
\|E_m(\vec t)-E_{\tilde W}(\vec t)\|_{1}\leq 
8|\sigma||A|^2k^{14\delta}(2^6|\sigma ||A|^2k^{1+5\delta})^{m}<k^{-1+9\delta}(2^6|\sigma ||A|^2k^{1+5\delta})^{m}. \label{mm++}
\end{equation} 
The projection $E_{\tilde W}(\vec t)$ is given by the series  \eqref{3.2.14}, \eqref{3.2.10a} with $\hat W=\tilde W-\sum _{q\in \Gamma (R_0)}P_qV_qP_q$ instead of $\hat W_0$. The series converges  with respect to $\|\cdot \|_1$:
\begin{equation}\label{ii-b}
\|\hat G_{r}(k,\vec t)\|_{1}\leq 2k^{-1+8\delta},\ \ \mbox{when  } r<20, \end{equation}
\begin{equation}\label{ii-bc}
\|\hat G_{r}(k,\vec t)\|_{1}\leq 2k^{-r/20},\ \ \mbox{when  } r\geq 20. \end{equation}

\end{lemma}
\begin{proof}
Let  $B(z)$ be given by \eqref{M17b} with $\hat W$ instead of $\hat W _s$.  Obviously, $B(z)$ is the limit of $B_m(z)$  in $\|\cdot \|_1$-norm. The estimate \eqref{M17c} yields:
\begin{equation} \label{M17c-1}
\|B(z)\|_1 \leq 2k^{2\delta},\ \  \|B^3(z)\|_1 \leq 2k^{-\frac{1}{5}+21\delta },  \ \ \ \            z \in C_0.
\end{equation}
It follows that 
 that $E(t)$ admits the expansion  \eqref{3.2.14}, \eqref{3.2.10a}. To obtain \eqref{ii-b} and \eqref{ii-bc} we sum up the estimates \eqref{estimate of difference of G} and use \eqref{Ap3a} and \eqref{Ap3b}, \eqref{M19b}.  Obviously, $\hat G_r$ corresponding to $\tilde W$ is the limit of  $\hat G_{m,r}$ in $\|\cdot\|_1$ norm.
Summing the estimates \eqref{estimate of difference of E-a}, we obtain \eqref{mm++}.

\end{proof}

%
%
%
\begin{definition} \label{def2} Let $u(\vec x)$ be defined as in Corollary  \ref{def1} for the potential  $\tilde{W}(\vec x)$. Let $\psi (\vec x)$ be the periodic part of $u(\vec x)$.\end{definition}

The next lemma follows from the estimate \eqref{mm++}.

\begin{lemma} \label{lemma3 2l>n'}
Suppose $\vec t$ belongs to the $(k^{-2-2\delta })$-neighborhood in $K$ 
of the 
non-resonant set $\chi _3(k, V,\delta )$. Then for every sufficiently 
large $k>k_{1}(V,\delta )$ and every $A\in \C: |\sigma||A|^2 <k^{-1-6\delta }$, the sequence $\psi _m(\vec x)$ converges to the function $\psi (\vec x)$ with respect to $\|\cdot \|_*$:
\begin{equation}
\|\psi_{{m}}-\psi \|_{*}\leq |A|k^{-1+9\delta}(2^6|\sigma ||A|^2k^{1+5\delta})^{m}. \label{mm+++}
\end{equation}
\end{lemma}

\begin{corollary} \label{3.7} The sequence
$u_{{m}}$ converges to $u_{}$ in $L^{\infty}(Q)$.\end{corollary}
\begin{corollary}\label{AW=W 2l>n} 
$${\mathcal M}W=W.$$
\end{corollary}
\begin{proof}[Proof of Corollary \ref{AW=W 2l>n}] 
Considering as in (\ref{**}), we obtain:
\begin{equation}
\big\|{\mathcal M}W_{m}-{\mathcal M}W_{}\big\|_{*} 
\leq |\sigma|\big\|\psi_{{m}}-\psi \big\|_{*}\big\|\bar{\psi}_{{m}}+\bar{\psi}\big\|_{*},\label{****}
\end{equation}
It immediately follows from Lemma \ref{lemma3 2l>n'} that ${\mathcal M}W_{m}\rightarrow {\mathcal M}W $
 with respect to $\|\cdot \|_*$.
Now, by (\ref{def of successive sequence A 4l>n+1}) and \eqref{****}, we have  ${\mathcal M}W=W$.
\end{proof} 
Let $\lambda_{{m}}(\vec t)$,  
corresponding to 
$\tilde W_m$. By perturbation theory, they have a limit $\lambda _{\tilde W}(\vec t)$, which is an eigenvalue of  $H_0+\tilde W$. This eigenvalue is unique in the interval $(k^2-k^{-1-\delta},k^2+k^{-1-\delta})$.

\begin{lemma} \label{lemma lambda 2l>n} Under conditions of Lemma \ref{lemma3 2l>n}
the sequence
 $\lambda_{{m}}(\vec t)$ converges to $\lambda_{\tilde W}(\vec t)$ being given by \eqref{3.2.13} and
 \begin{equation}
|\hat g_r(k,\vec t)|<2k^{-2+80\delta},\ 2\leq r<40,
\label{gr}
\end{equation}
\begin{equation}
|\hat g_r(k,\vec t)|<2k^{-r/20},\ r\geq 40.
\label{gr1}
\end{equation}
where $\hat g_r$ are given by  \eqref{3.2.10} with $\hat W=\tilde W- \sum _{q\in \Gamma (R_0)}P_qV_qP_q$ instead of $\hat W_0$.
\end{lemma}
\begin{proof}
Let us show that the series  \eqref{3.2.13} converges. Notice that $\hat g_1=0$. We consider $\hat g_r - \hat g_{0,r}$ where index $0$ stands for potential $V$. To estimate $\|B(z)^r - B_0(z)^r\|_1$, $r\geq40$, we follow \eqref{estimate of difference of B}. Instead of \eqref{B} we now write
\begin{equation}\label{Bnew}
 \max_{z\in C_{0}}\|B_{s}^r(z)-B^r_{s-1}(z)\|_{\hbox{\bf S}_1}<2^rk^{1 +\delta }\|\tilde{W}_{s}-\tilde{W}_{s-1}\|_{*}k^{-(\frac{1}{5}-21\delta )\left[\frac{r-3}{3}\right]+4\delta }k^{2+2\delta}.\end{equation}
 Here we twice used \eqref{lS2} instead of \eqref{M17c-1} worsening the estimate, but improving the class to $\hbox{\bf S}_1$. 
 Summing up in $s$ and using \eqref{M19b} we obtain (recall $|\sigma||A|^2\leq k^{-1-6\delta}$)
\begin{equation}\label{g-g}
|\hat g_r(k,\vec t)-\hat g_{0,r}(k,\vec t)|\leq 2^{r+2}k^{4\delta}k^{-(\frac{1}{5}-21\delta )\left[\frac{r-3}{3}\right]}\leq k^{-r/20},\ \ r\geq40.
\end{equation}
For $2\leq r< 40$ we apply a little bit different argument. Namely, Let us consider two projections ${\mathbb E}_0={\mathbb E}_j$, ${\mathbb E}_1=I-{\mathbb E}_j$, here ${\mathbb E}_j$ is the spectral projection of $\hat H$. Note that
$$\oint _{C_0} \left({\mathbb E}_1B(z){\mathbb E}_1\right)^r dz=0,\ \ r=1,2,...,$$
since the integrand is holomorphic inside $C_0$. Hence,
\begin{align}\label{LAMBDA}
\nonumber& \oint _{C_0} B(z)^r dz =\oint _{C_0} \big(B(z)^r - \left({\mathbb E}_1B(z){\mathbb E}_1\right)^r\big)dz=\\
\nonumber& \sum _{i_1,...,i_{r+1}=0,1, \exists s: i_s=0}\oint _{C_0 }{\mathbb E}_{i_1}B(z){\mathbb E}_{i_2}B(z)....{\mathbb E}_{i_r}B(z){\mathbb E}_{i_{r+1}}dz. \end{align}

Obviously, ${\mathbb E}_{i_1}B(z){\mathbb E}_{i_2}B(z)....{\mathbb E}_{i_r}B(z){\mathbb E}_{i_{r+1}}$ is in the trace class $\hbox {\bf S}_1$ if at least one index $i_s$,  $1\leq s\leq r+1$ is zero, since ${\mathbb E}_0 \in \hbox {\bf S}_1$. 
It follows:
\begin{align}
\nonumber&\|{\mathbb E}_{i_1}B(z){\mathbb E}_{i_2}B(z)....{\mathbb E}_{i_r}B(z){\mathbb E}_{i_{r+1}}-{\mathbb E}_{i_1}B_0(z){\mathbb E}_{i_2}B_0(z)....{\mathbb E}_{i_r}B_0(z){\mathbb E}_{i_{r+1}}\|_{\hbox {\bf S}_1}\leq 
\\ \nonumber& 2^r\|B\|^{r-1}\|B-B_0\|_1\leq k^{2r\delta }k^{1+\delta}\|\tilde W-V\|_*\leq k^{2r\delta-1-\delta}.
\end{align}
Thus,
\begin{equation}\label{g-g1}
|\hat g_r(k,\vec t)-\hat g_{0,r}(k,\vec t)|\leq k^{-2+80\delta},\ \ 2\leq r<40.
\end{equation}
Now (see \eqref{g-g}, \eqref{g-g1}, \eqref{3.2.15}, \eqref{3.2.21} and \eqref{3.2.24}), \eqref{gr} and \eqref{gr1} follow. 

\end{proof}

Considering as in the proof of Theorem 4.3 and Corollary 4.3 from \cite{K97}, one can prove 

\begin{theorem} \label{2.3-a} 
Suppose $\vec t$ belongs to the $(k^{-2-2\delta })$-neighborhood in $K$ 
of the 
non-resonant set $\chi _3(k,V,\delta )$. Then for every sufficiently 
large $k>k_{1}(V,\delta)$ and every $A\in \C: |\sigma||A|^2 <k^{-1-6\delta }$, the series (\ref{3.2.13}), (\ref{3.2.14}) for the potential $\tilde W$
can be differentiated with respect to $\vec t$ any number of times, and 
they retain their asymptotic character. Coefficients $\hat g_r(k,\vec t)$ and 
$\hat G_r(k,\vec t)$ satisfy the following estimates in the $(k^{-2-2\delta })$-neighborhood in $\C^3$ of the nonsingular set $\chi _3(k,V,\delta )$:
\begin{equation}
\mid T(m)\hat g_r(k,\vec t)\mid <m!k^{-2+80\delta}k^{\mid m\mid (2+2\delta )},\ \ 2\leq r<40, \label{2.2.32a-1}
\end{equation}
\begin{equation}
\mid T(m)\hat g_r(k,\vec t)\mid <m!k^{-r/20}k^{\mid m\mid (2+2\delta )},\ \ r\geq 40, \label{2.2.32a-2}
\end{equation}
%
%
\begin{equation}
\| T(m)\hat G_r(k,\vec t)\|_{1}<m!k^{-1+8\delta}k^{\mid m\mid (2+2\delta )},\ \ 1\leq r<20,  \label{2.2.33a-1}
\end{equation}
\begin{equation}
\| T(m)\hat G_r(k,\vec t)\|_{1}<m!k^{-r/20}k^{\mid m\mid (2+2\delta )},\ \ r\geq20. \label{2.2.33a-2}
\end{equation}
\end{theorem}

\begin{corollary} \label{derivatives-1} There are the estimates for the perturbed eigenvalue and its
spectral projection:
\begin{equation}
\left| T(m)\big(\lambda_{\tilde W} (\vec t)-p_j^{2}(\vec t)\big)\right| <m!k^{-2+80\delta}k^{\mid m\mid (2+2\delta )},\label{2.2.34b-1}
\end{equation}
\begin{equation} 
\| T(m)(E_{\tilde W}(\vec t)-{\mathbb E}_j)\|_{1} <m!k^{-1+8\delta}k^{\mid m\mid (2+2\delta )}.\label{2.2.35-1}
\end{equation}
In particular,
\begin{equation}
\left| \lambda _{\tilde W}(\vec t)-p_j^{2}(\vec t)\right| <k^{-2+80\delta}, \label{2.2.34b-3}
\end{equation}
\begin{equation} 
\| E_{\tilde W}(\vec t)-{\mathbb E}_j\|_{1} <k^{-1+8\delta}, \label{2.2.35-3}
\end{equation}
\begin{equation}
\left|\nabla\lambda _{\tilde W}(\vec t)-2 \vec p_j(\vec t)  \right|<k^{82\delta}.\label{2.2.34b-2}
\end{equation}
\end{corollary}

We have the following main result for the nonlinear equation with quasi-periodic condition. 

\begin{theorem}  \label{main theorem 4l>n+1}
Suppose $\vec t$ belongs to the $(k^{-2-2\delta })$-neighborhood in $K$ 
of the 
non-resonant set $\chi _3(k,V,\delta )$,  $k>k_{1}(V,\delta)$ and $A\in \C: |\sigma||A|^2 <k^{-1-6\delta }$. Then,
there is  a function $u(\vec{x})$, depending on $\vec t $ as a parameter, and a real value $\lambda(\vec t)$,
satisfying the  equation 
\begin{align}\label{main equation 4l>n+1}
-\Delta u(\vec{x})+V(\vec{x})u(\vec{x})+\sigma |u(\vec{x})|^{2}u(\vec{x})=\lambda u(\vec{x}),~\vec{x}\in 
Q,
\end{align}
and  the quasi-periodic boundary condition (\ref{main condition, 4l>n+1}). The following  formulas hold:
\begin{align} \label{solution construction 2l>n}
u( \vec{x})=&Ae^{i\langle{ \vec{p}_{j}(\vec t), \vec{x} }\rangle}\left(1+\tilde{u}( \vec{x})\right),\\ 
\lambda(\vec{t})=&p_j^{2}(\vec t)+\sigma|A|^2+O\left( \left(k^{-1+72\delta}+\sigma |A|^2 \right) k^{-1+8\delta}\right),~ \label{kkk}
\end{align} 
where $\tilde{u}( \vec{x})$ is periodic and \begin{equation}
\|\tilde{u}\|_{*}\leq k^{-1+8\delta}. \label{June7}
\end{equation}
\end{theorem}
\begin{proof}
Let us consider the function $u$ given by Definition \ref{def2} and  the value $\lambda_{\tilde{W}}(\vec t)$.  
They solve the equation
\begin{align}\label{semi solution 2l>n}
-\Delta u_{}(\vec{x})+\tilde{W}(\vec{x})u_{}(\vec{x})=\lambda_{\tilde{W}}(\vec t) u_{}(\vec{x}),\ \ \  \vec{x}\in Q,
\end{align}
and $u$ satisfies the quasi-boundary condition (\ref{main condition, 4l>n+1}).
By Corollary \ref{AW=W 2l>n}, we have
$$W(\vec{x})={\mathcal M}W(\vec{x})=V(\vec{x})+\sigma|u(\vec{x})|^{2}.$$
Hence,
\begin{align}\label{tilde W 2l>n}
\nonumber\tilde{W}(\vec{x})=&W(\vec{x})-\frac{1}{(2\pi)^{3}}\int_{Q}W(\vec{x})d\vec{x}
=V(\vec{x})+\sigma|u(\vec{x})|^{2}-\sigma\|u\|_{L^{2}(Q)}^{2}.
\end{align}
Substituting the last expression into \eqref{semi solution 2l>n}, we obtain that  $u( \vec{x})$
satisfies \eqref{main equation 4l>n+1} with 
\begin{equation}\label{lambda06}
\lambda(\vec t)=\lambda_{\tilde{W}}(\vec t)+\sigma\|u\|_{L^{2}(Q)}^{2}
=\lambda_{\tilde{W}}(\vec t)+\sigma |A|^2\sum _{q\in \Z^3} \big|\left(E_{\tilde W}\right)_{qj}\big|^2=\lambda_{\tilde{W}}(\vec t)+\sigma |A|^2\left(E_{\tilde W}\right)_{jj}.\end{equation}
Note that $(\hat G_1)_{jj}=0$ and, therefore, $\left(E_{\tilde W}\right)_{jj}=1+O(k^{-(1-8\delta)})$.
%
%
Further, by the definition of $u( \vec{x})$, we have
\begin{align}\label{u06}
&u( \vec{x}):=Ae^{i\langle{ \vec{p}_{j}(\vec t), \vec{x} \rangle}}\sum\limits_{q\in \Z^3}\left(E_{\tilde W}\right)_{q+j,j}e^{i\langle{\vec p_q(0), \vec{x} \rangle}}.
\end{align} 
Using \ formulas \eqref{lambda06} and \eqref{u06} and estimates \eqref{2.2.34b-3} and \eqref{2.2.35-3}, we obtain (\ref{solution construction 2l>n}) and \eqref{June7}, respectively.
\end{proof}

\begin{lemma}\label{L:2.12} For any sufficiently large $\lambda $, every $A\in \C: |\sigma||A|^2 <k^{-1-6\delta }$, $\lambda=k^{2}$ and for every
$\vec{\nu}\in\mathcal{B} (\lambda)$, there is a
unique $\varkappa =\varkappa (\lambda , A, \vec{\nu})$ in the
interval $$
I:=[k-k^{-2-2\delta},k+k^{-2-2\delta}], $$ such that \begin{equation}\label{2.70}
\lambda (\varkappa \vec{\nu},A)=\lambda . \end{equation}
Furthermore, 
\begin{equation} \label{varkappa}
|\varkappa (\lambda , A, \vec{\nu}) - \tilde k| \leq C(||V||_*)\left(k^{-1+72\delta}+|\sigma ||A|^2\right)k^{-2+8\delta}, \ \ \tilde k=(\lambda -\sigma |A|^2)^{1/2}.\end{equation}
\end{lemma} \begin{proof} 
Taking into account  \eqref{formulaB} and using   formulas \eqref{2.2.34b-3}, \eqref{2.2.34b-2}  and Implicit Function Theorem,
we prove the lemma.  The proof is completely analogous to that for the linear case. \end{proof}

\begin{theorem} \label{iso} \begin{enumerate} \item For any sufficiently
large $\lambda $ and every $A\in \C: |\sigma||A|^2 <k^{-1-6\delta }$, the set $\mathcal{D}(\lambda , A)$, defined by \eqref{isoset} is a distorted
sphere with holes; it can be described by the formula
\begin{equation}  {\mathcal D}_{}(\lambda, A)=\{\vec k:\vec
k=\varkappa _{}(\lambda, A,\vec{\nu})\vec{\nu},
\ \vec{\nu} \in {\mathcal B}_{}(\lambda)\},\label{May20} \end{equation} where $
\varkappa (\lambda, A, \vec \nu)=\tilde k+h (\lambda, A, \vec \nu)$ and $h (\lambda, A, \vec \nu)$ obeys the
inequalities
\begin{equation}\label{2.75}
|h|<C(||V||_*)\left(k^{-1+72\delta}+|\sigma ||A|^2\right)k^{-2+8\delta}<C(||V||_*)k^{-2\gamma _1},
\end{equation}
with $2\gamma _1=3-80\delta >0$,
\begin{equation}\label{2.75'}\left|\nabla _{\vec \nu }h\right| <
C(||V||_*)k^{-2\gamma_1+2+2\delta}=C(||V||_*)k^{-1+82\delta}.
\end{equation}
\item The measure of $\mathcal{B}(\lambda)\subset S_{2}$ satisfies the
estimate
\begin{equation}\label{theta1}
L\left(\mathcal{B}\right)=\omega _{2}(1+O(\lambda^{-\delta })).
\end{equation}


\item The surface $\mathcal{D}(\lambda ,A)$ has the measure that is
asymptotically close to that of the whole sphere of the radius $k$ in the sense that
\begin{equation}\label{2.77}
\bigl |\mathcal{D}(\lambda,A )\bigr|\underset{\lambda \rightarrow
\infty}{=}\omega _{2}\lambda\bigl(1+O(\lambda^{-\delta})\bigr).
\end{equation}
\end{enumerate}
\end{theorem}

\begin{proof} Statements (2) and (3) are parts of Lemma~\ref{L:2.13a}. Statement (1) follows from Corollary~\ref{derivatives-1}.\end{proof}

\section{Appendices}
\subsection{Appendix 1. Proof of \eqref{l}}  By Lemma 4.12 in \cite{K97},
 \begin{align}\label{lll}
\max_{z\in C_{0}}\Big\|P(\hat H (\vec t)-z)^{-1/2}\Big\|<k^{1/10+\delta}.
\end{align}
 \begin{align}\label{ll}
\max_{z\in C_{0}}\Big\|(I-P)(\hat H (\vec t)-z)^{-1/2}\Big\|<k^{(1+\delta)/2},
\end{align}
here $P=\sum P_q$, the projectors $P_q$ being given by \eqref{3.2.7a}. Consider that each $P_q$ is a sum of orthogonal diagonal  projections $P_{qi}$, $P_q=\sum _iP_{qi}$,  and 
\begin{equation} \hat H=\sum _{q,i}P_{qi}\hat HP_{qi}, \label{rr} 
\end{equation}
the rank of each $P_{qi}$ not exceeding $ck^{1/5} $,  (see (4.3.22) in \cite{K97}). It follows:
 \begin{align}\label{llla}
\max_{z\in C_{0}}\Big\|P(\hat H (\vec t)-z)^{-1/2}\Big\|_1<ck^{3/10+\delta}.
\end{align} Since the operator in \eqref{ll} is diagonal, we obtain:
 \begin{align}\label{lla}
\max_{z\in C_{0}}\Big\|(I-P)(\hat H (\vec t)-z)^{-1/2}\Big\|_1<k^{(1+\delta)/2},
\end{align}
The last two estimates yield \eqref{l}.
\subsection{Appendix 2. Proof of \eqref{3.2.11b}} It is proven in \cite{K97}, Lemma 4.14:
\begin{equation}
\|B_0\|<k^{2\delta },\ \ \|B_0^3\|<k^{-1/5+21\delta }. \label{3.2.11bb}
\end{equation}
Now we show how to modify it to a slightly stronger estimate  \eqref{3.2.11b}. Indeed, let $B_0^{(1)}=(I-P)B_0(I-P)$. It is proven in \cite{K97} that
\begin{equation}
\|B_0^{(1)}\|_1<k^{2\delta },\ \ \|(B_0^{(1)})^3\|_1<k^{-1/5+20\delta }. \label{3.2.11bbb}
\end{equation}
Clearly, to obtain \eqref{3.2.11b} from \eqref{3.2.11bbb} it suffices to show that
\begin{equation}
\|B_0-B_0^{(1)}\|_1<k^{-1/5+12\delta }. \label{3.2.11bbbb}
\end{equation}
By the definition of $B_0$, $PB_0P$=0. Let us consider $B_0^{(2)}=PB_0(I-P)$ and  diagonal projections $\hat P$, $\check P$ :
\begin{equation} 
(\hat P)_{jj}=\left\{ \begin{array}{ll}1, &\mbox{if $j : |p_j(t)^2-k^2|<k^{3/5-\delta};$
}\\0, &\mbox{otherwise.}\end{array}\right.  \label{3.2.7ab}
\end{equation}
\begin{equation} 
(\check P)_{jj}=\left\{ \begin{array}{ll}1, &\mbox{if $j : |p_j(t)^2-k^2|<k^{-1/5-6\delta};$
}\\0, &\mbox{otherwise.}\end{array}\right.  \label{3.2.7abc}
\end{equation}
The definition of $P$, in particular,  exclusion of the set $T$ in \eqref{3.2.7a}, yields:
$\hat PB_0^{(2)}\hat P=0$.  It follows from the definitions of $\chi _3(V,\delta )$ (see  (4.3.40) in \cite{K97}) and the projector $P$ that
$B_0\check P=0$.
 Now, considering \eqref{lll}, we obtain the inequality for matrix elements: $\|(B_0^{(2)})_{jl}\|<\|V\|k^{-2/5+7\delta }$. Taking into account \eqref{rr}, we get:
\begin{equation}
\|B_0^{(2)}\|_1<c\|V\|k^{-1/5+7\delta }. \label{3.2.11bbbbb}
\end{equation}
The analogous estimate holds for $B_0^{(3)}=(I-P)B_0P$. Now \eqref{3.2.11bbbb} easily follows from \eqref{3.2.11bb} and \eqref{3.2.11bbbbb}.
\subsection{Appendix 3. Proof of \eqref{Ap3a}, \eqref{Ap3b}} Estimate \eqref{Ap3a}
 is proven in \cite{K97}. The estimate \eqref{Ap3b} follows from \eqref{l} and \eqref{3.2.11b}.

\


\begin{thebibliography}{99}
























\bibitem[1]{K97} Yulia E. Karpeshina, {\it Perturbation Theory for the Schr\"odinger Operator with a Periodic Potential}, Springer, 1997.

\bibitem[2] {KS17}  Yulia Karpeshina, Seonguk Kim, R. Shterenberg, {\it Solutions of Nonlinear Polyharmonic Equation with Periodic Potentials}, 
Analysis as a Tool in Mathematical Physics: In Memory of Boris Pavlov (Operator Theory: Advances and Applications,  276), pp 401-- 416, 2020.


\bibitem[3]{KS02} V. V. Konotop and M. Salerno, {\it Modulational instability in Bose-Einstein condensates in optical lattices}, Phys. Rev. A 65, 021602, 2002.



\bibitem[4]{LO03} Pearl J.Y. Louis, Elena A. Ostrovskaya, Craig M. Savage and Yuri S. Kivshar, {\it Bose-Einstein Condensates in Optical Lattices: Band-Gap Structure and Solitons}, Phys. Rev. A 67, 013602, 2003.


\bibitem[5]{PS08} C. J. Pethick, H. Smith, {\it Bose-Einstein Condensation in Dilute Gases}, Cambridge, 2008.


\bibitem[6]{YB13} A. V. Yulin, Yu. V. Bludov, V. V. Konotop, V. Kuzmiak, and M. Salerno, {\it Superfluidity breakdown of periodic matter waves in quasi-one-dimensional annular traps via resonant scattering with moving defects}, Phys. Rev. A 87, 033625 -- Published 25 March 2013.

\bibitem[7]{YD03} Alexey V. Yulin and Dmitry V. Skryabin, {\it Out-of-gap Bose-Einstein solitons in optical lattices}, Phys. Rev. A 67, 023611, 2003.

\bibitem[8]{2r} Kittel, Ch., {\em Introduction to Solid State Physics.} New-York, Wiley, c1976.

\bibitem[9]{1r} Madelung, O., {\em Introduction to Solid State Theory.} Berlin, New-York, Springer-Verlag, 1978.

\bibitem[10]{3r} Ziman, J.M. {\em Principles of the Theory of Solids.} Cambridge, University Press, 1965.

\end{thebibliography}
\end{document}